\documentclass{article}
  \usepackage{geometry}
\usepackage{amsmath}
\usepackage{amsthm}
\usepackage{amssymb}

\newtheorem{theorem}{Theorem}[section]
\newtheorem{lemma}[theorem]{Lemma}

\theoremstyle{definition}

\theoremstyle{corollary}
\newtheorem{corollary}[theorem]{Corollary}

\theoremstyle{remark}

\numberwithin{equation}{section}
\newtheorem{problem}{Problem}[section]
\newcommand{\br}{\mathbb R}

\title{A Heat Conduction Problem with Sources Depending on the
Average of the Heat Flux on the Boundary
        }
\author{
{Mahdi Boukrouche\thanks{ Address : Lyon University, F-42023
Saint-Etienne, Institut Camille Jordan CNRS UMR 5208, 23 rue  Paul
Michelon 42023 Saint-Etienne Cedex 2, France.
Mahdi.Boukrouche@univ-st-etienne.fr}} \and 
Domingo A. Tarzia\thanks{
Address: Departamento de Matem\'atica-CONICET, FCE, Univ. Austral,
Paraguay 1950, S2000FZF Rosario, Argentina. DTarzia@austral.edu.ar}
}

\date{}

\begin{document}
\maketitle \normalsize

\begin{abstract}
Motivated by the modeling of temperature regulation in some mediums,
we consider the non-classical heat conduction equation in the
 domain $D=\br^{n-1}\times\br^{+}$  for which the
internal energy supply depends  on an average in the time variable
of  the
heat flux  $(y, s)\mapsto V(y,s)= u_{x}(0 , y , s)$  on the boundary $S=\partial D$.
The solution to the
problem is found for an integral  representation depending on the
heat flux
 on $S$ which is an additional unknown  of  the considered problem.
We obtain
that the heat flux $V$ must satisfy a Volterra integral equation of second kind in the
time variable $t$ with a parameter in $\br^{n-1}$. Under some conditions on data,
we show that  a unique local solution  exists,  which  can
be extended globally in time. Finally in the one-dimensional case, we obtain the explicit solution by using
the Laplace transform and the Adomian decomposition method.
\bigskip


\smallskip

\end{abstract}

\maketitle
\pagestyle{myheadings}
\thispagestyle{plain}
\markboth{Problem with Sources Depending on the
Average of the Heat Flux on the Boundary}{Mahdi Boukrouche  and Domingo A. Tarzia}

\section{Introduction}

Let consider the domain $D$ and its boundary $S$
  defined by
      \begin{eqnarray}
      &&D= \br^{n-1}\times\br^{+}=\{ (x , y)\in \br^{n} : \quad x=
      x_{1} >0, \quad y=( x_{2}, \cdots, x_{n})\in \br^{n-1}\}, \qquad
      \\
      &&S=\partial D=\br^{n-1}\times\{0\}=  \{(x , y)\in \br^{n} : \quad
      x=0, \quad y \in \br^{n-1} \}.
    \end{eqnarray}

The aim of this paper is to study the  following  problem
\ref{pbb} on
 the non-classical heat equation, in the semi-n-dimensional space domain
$D$ with non local sources, for which the internal energy supply
depends on the average
 $\frac{1}{t}\int_{0}^{t} u_{x}(0 , y , s) ds$
 of  the  heat flux on the
boundary $S$.

 \begin{problem}\label{pbb}
 Find the temperature $u$ at $(x , y , t)$ satisfying the
 following conditions
 \begin{eqnarray*}\label{eqchN}
 u_{t} - \Delta u &=& -F\left
 (\frac{1}{t}\int_{0}^{t}
 u_{x}(0, y , s)ds\right
 ), \qquad x>0,
    \quad y\in \br^{n-1}, \quad  t>0, \label{cNpb}\nonumber\\
            u(0, y ,  t)&=& 0, \quad y\in\br^{n-1},   \quad t>0, \nonumber\\
  u(x, y, 0)&=& h(x , y),   \qquad x>0,  \quad y\in\br^{n-1},\label{cNIpb}
    \end{eqnarray*}
\end{problem}
where $\Delta$ denotes the Laplacian in $\br^{n}$. This problem is
motivated by modeling the temperature in an isotropic medium with the
average of non-uniform and non local sources
 that provide cooling or heating system, according to
the properties of the function  $F$ with respect to the heat flow $(y, s)\mapsto V(y, s)=u_{x}(0, y , s)$
 at the  boundary $S$,  see
    \cite{cannon1984, carslaw59}.
Some references on the subject are \cite{MT-qam} where
$F\left(\frac{1}{t}\int_{0}^{t} u_{x}(0 , y , s) ds\right)$ is
replaced by $F(u_{x}(0 , y , t))$, or \cite{MB-DT1} where  is
replaced by $F\left(\int_{0}^{t} u_{x}(0 , y , s) ds\right)$;  see
also \cite{berrone}, \cite{CeTaVi2015}, \cite{tarzia-villa},
\cite{villa} where the semi-infinite case of this nonlinear problem
with $F(u_{x}(0 , y , t))$ have been considered. The non-classical
one-dimensional heat equation in a slab with fixed or moving
boundaries was studied in \cite{CeTaVi2015}, \cite{salva}. See also
other references on the subject
\cite{bor-dt2006}-\cite{bor-dt2010-2},
 \cite{cannon1989},
\cite{GLASHOFF81}-\cite{KENMOCHI88}. To our knowledge, it is the
first time that the solution to the average of a non-classical heat
conduction of the type of  Problem \ref{pbb} is given. Other
non-classical problems can be found in \cite{BBKW2010}.

In \cite{MT-qam} basic solution to the n-dimensional heat
equation, and a technical Lemma was established.
 We prove in Section 2  the local existence of a
solution  for the considered Problem \ref{pbb}
 under some conditions on data $F$ and $h$ which can be extended globally in times.
 Moreover, in Section 3 we consider the corresponding one dimensional problem
 and we obtain its explicit solution for the heat flux and the average of the total
 heat flux at the face $x=0$, by using the Laplace transform and also the Adomian decomposition method \cite{Adm, Adl, Adm1, MB-DT1, waz1,
 waz3}.



 %
%
%
 %
%
%
%
%

 \bigskip

  \section{Existence results }

In this Section, we give first in Theorem \ref{th2.1}, the integral
representation (\ref{intSol}) of the solution of the considered
Problem \ref{pbb}, but it depends on the heat flow $V $on the
boundary $S$, which
satisfies the  Volterra integral equation (\ref{Volera}) with
initial condition (\ref{ci}). Then we prove, in Theorem \ref{th3.2},
under some assumptions on the data, that there exists a unique
solution of the problem locally in times which can be
extended globally in times.

We first recall here
the Green's function for the n-dimensional heat
equation   with homogenuous Dirichlet's boundary conditions, given
the following expression

\begin{eqnarray}\label{G}
        G_{1}(x , y , t ; \xi , \eta , \tau)=
{\exp\left[-{\|y-\eta\|^{2}
                     \over 4(t-\tau)}\right]
              \over  \left(2\sqrt{\pi(t-\tau)}\right)^{n-1}}G(x , t , \xi,
\tau),
            \end{eqnarray}
   where 
   $G$ is the Green's function for
   the one-dimensional case given by
$$G(x , t , \xi, \tau)=  {e^{-{(x-\xi)^{2} \over 4(t-\tau)}}-
e^{-{(x+\xi)^{2}\over 4(t-\tau)}}\over
    2\sqrt{\pi(t-\tau)}} \qquad t>\tau.
    $$

  \begin{theorem}\label{th2.1}
  The integral representation of a solution of the
 Problem {\rm\ref{pbb}} is given by the following expression
\begin{eqnarray}\label{intSol}
    && u(x , y ,  t)= u_{0}(x , y , t)\nonumber\\
     &&-\int_{0}^{t}{{\rm erf}\left({x\over
     2\sqrt{t-\tau}}\right)  \over (2\sqrt{\pi(t-\tau)})^{n-1}}
      \left[  \int_{\br^{n-1}} e^{\left[-{\|y-\eta\|^{2}\over
      4(t-\tau)}\right]}
      F\left(
 {1\over \tau}\int_{0}^{\tau}V(\eta\ , s) ds\right)
      d\eta\right]d\tau \quad
      \end{eqnarray}
  where
$
\zeta \mapsto {\rm erf}\left(\zeta\right)=\left({2\over \sqrt{\pi}}
\int_{0}^{\zeta}e^{-X^{2}} dX \right)
$
is the error function,
\begin{eqnarray}\label{ci}
u_{0}(x , y , t)=\int_{D} G_{1}(x ,
y , t ; \xi , \eta , 0) h(\xi ,
     \eta) d\xi d\eta
   \end{eqnarray}
and the heat flux $(y, t)\mapsto V(y , t)=u_{x}(0 , y , t)$ on the  surface $x=0$,
  satisfies the following Volterra integral equation
         \begin{eqnarray}\label{Volera}
        && V(y , t)=  V_{0}(y , t)
         \nonumber\\
         &&-  2\int_{0}^{t}  {1\over  (2\sqrt{\pi(t-\tau)})^{n}}
          \left[  \int_{\br^{n-1}} e^{\left[-{\|y-\eta\|^{2}\over
          4(t-\tau)}\right]}
           F\left(
 {1\over \tau}\int_{0}^{\tau}V(\eta\ , s) ds\right)
          d\eta\right]d\tau \qquad
          \end{eqnarray}
     in the variable $t>0$, with $y \in \br^{n-1}$ is a parameter
     where
\begin{eqnarray}\label{ci}
V_{0}(y , t)= \int_{D} G_{1,x}(0 , y , t ; \xi , \eta , 0) h(\xi ,
         \eta) d\xi d\eta.
     \end{eqnarray}
  \end{theorem}
  \begin{proof}
As  the boundary condition in Problem {\rm(\ref{pbb})} is
homogeneous,  we have from \cite{FRIEDMAN, LADY68}

\begin{eqnarray}\label{uu}
 u(x , y , t)&=&
  \int_{D} G_{1}(x , y , t ; \xi ,\eta , 0) h(\xi ,\eta) d\xi d\eta
     \nonumber\\
     &&-
     \int_{0}^{t}
     \int_{D} G_{1}(x , y , t ; \xi , \eta ,
     \tau) F\left(
 {1\over \tau}\int_{0}^{\tau}V(\eta\ , s) ds\right) d\xi d\eta d\tau,
           \end{eqnarray}
and therefore
\begin{eqnarray}\label{a}
 u_{x}(x , y , t)&=&
  \int_{D} G_{1,x}(x , y , t ; \xi ,\eta , 0) h(\xi ,\eta) d\xi d\eta
     \nonumber\\
     &&-
     \int_{0}^{t}
     \int_{D} G_{1,x}(x , y , t ; \xi , \eta ,
     \tau) F\left(
 {1\over \tau}\int_{0}^{\tau}V(\eta\ , s) ds\right)
      d\xi d\eta d\tau.
           \end{eqnarray}

From    {\rm(\ref{G})} (the definition of $G_{1}$)   by
derivation     with respect to $x$, taking $x=0$ we obtain

\begin{eqnarray}\label{a3}
&&\int_{D}G_{1,x}(0 , y , t ; \xi , \eta ,\tau)
F\left( {1\over \tau}\int_{0}^{\tau}V(\eta\ , s) ds\right)d\xi
d\eta =
\nonumber\\
&&\int_{\br^{n-1}}{F\left(
 {1\over \tau}\int_{0}^{\tau}V(\eta\ , s) ds\right)
 e^{-{\|y-\eta\|^{2}\over
4(t-\tau)}} \over(t-\tau)^{{n+2\over 2}}(2\sqrt{\pi})^{n}}
 \left(\int_{0}^{+\infty}\xi e^{-{\xi^{2}\over 4(t-\tau)}}
d\xi\right)d\eta \nonumber\\
 &&={2\over (2\sqrt{\pi(t-\tau)})^{n}}
 \int_{\br^{n-1}}
 F\left(
 {1\over \tau}\int_{0}^{\tau}V(\eta\ , s) ds\right)
 e^{-{\|y-\eta\|^{2}\over
4(t-\tau)}} d\eta.\quad
 \end{eqnarray}
Thus taking $x=0$ in {(\ref{a})} with {(\ref{a3})}  we get
 {(\ref{Volera})}.

\bigskip

Also by (\ref{G}) we obtain
\begin{eqnarray*}
&&\int_{D}G_{1}(x , y , t ; \xi , \eta
,\tau)
F\left(
 {1\over \tau}\int_{0}^{\tau}V(\eta\ , s) ds\right)
d\xi d\eta
=
\nonumber\\
&&
{1\over
(2(\sqrt{\pi(t-\tau)})^{n}}\times
\int_{D} e^{{-\|y-\eta\|^{2}\over
4(t-\tau)}}\left[e^{-{(x-\xi)^{2}\over 4(t-\tau)}}- e^{-{(x+\xi)^{2}\over W
4(t-\tau)}}\right]
F\left(
 {1\over \tau}\int_{0}^{\tau}V(\eta\ , s) ds\right)
d\xi d\eta =
\nonumber\\
&&
{1\over (2(\sqrt{\pi(t-\tau)})^{n}}\int_{\br^{+}}
\left[e^{-{(x-\xi)^{2}\over 4(t-\tau)}}- e^{-{(x+\xi)^{2}\over 4(t-\tau)}}\right]d\xi \int_{\br^{n-1}}
e^{{-\|y-\eta\|^{2}\over 4(t-\tau)}}
F\left(
 {1\over \tau}\int_{0}^{\tau}V(\eta\ , s) ds\right)
d\eta
\end{eqnarray*}
 using
\begin{eqnarray*}\label{eqA}
\int_{0}^{+\infty} e^{-(x-\xi)^{2}\over 4(t-\tau)} d\xi &=&
2\sqrt{t-\tau} \left(\int_{-\infty}^{0} e^{-X^{2}} dX
+\int_{0}^{{x\over 2\sqrt{t-\tau}}} e^{-X^{2}} dX  \right)
\nonumber\\
&=& \sqrt{\pi(t-\tau)}\left(1+  {\rm erf}\left({x\over
2\sqrt{t-\tau}}\right)\right)
\end{eqnarray*}
and
\begin{eqnarray*}\label{eqB}
\int_{0}^{+\infty} e^{-(x+\xi)^{2}\over 4(t-\tau)} d\xi &=&
2\sqrt{t-\tau} \left(\int_{0}^{+\infty} e^{-X^{2}} dX
 -\int_{0}^{{x\over 2\sqrt{t-\tau}}} e^{-X^{2}} dX  \right)
\nonumber\\
&=& \sqrt{\pi(t-\tau)}\left(1-  {\rm erf}\left({x\over
2\sqrt{t-\tau}}\right)\right)
\end{eqnarray*}
 so we get
\begin{eqnarray*}\label{GK1}
&&\int_{D}G_{1}(x , y , t ; \xi , \eta ,\tau)F\left(
 {1\over \tau}\int_{0}^{\tau}V(\eta\ , s) ds\right)
d\xi
d\eta =
\nonumber\\
&&
{ {\rm erf}\left({x\over 2\sqrt{t-\tau}}\right) \over
         (2\sqrt{\pi(t-\tau)})^{n-1}} \int_{\br^{n-1}}
e^{-{\|y-\eta\|^{2}\over 4(t-\tau)}}
F\left(
 {1\over \tau}\int_{0}^{\tau}V(\eta\ , s) ds\right)
  d\eta.
\end{eqnarray*}
Taking this formula in (\ref{uu}) we obtain (\ref{intSol}).
\end{proof}

\begin{lemma}\label{lem3.2}
The simplified form of Volterra integral equation
{\rm(\ref{Volera})} is given by
\begin{eqnarray}\label{SV}
&&V(y , t)= {1\over t(2\sqrt{\pi\, t})^{n} } \int_{\br^{+}}\xi
e^{-{\xi^{2}\over 4t}}\left(\int_{\br^{n-1}} e^{-{\|y-
\eta\|^{2}\over 4t}}h(\xi ,
\eta)d\eta\right)d\xi\nonumber\\
&&-{2\over (2\sqrt{\pi})^{n} } \int_{0}^{t}
    {1   \over (t-\tau)^{n/2}}
\int_{\br^{n-1}}
F\left(
 {1\over \tau}\int_{0}^{\tau}V(\eta\ , s) ds\right)
e^{-{\|y- \eta\|^{2}\over
4(t-\tau)}}d\eta d\tau.
 \end{eqnarray}
\end{lemma}
\begin{proof}
Using the derivative, with respect to $x$, of {\rm(\ref{G})}, then
taking $x=0$  and $\tau=0$, then    taking the new expression of
 $V_{0}(y,t)$  in the Volterra integral
equation {\rm(\ref{Volera})} we obtain  {\rm(\ref{SV})}.
  \end{proof}

\begin{theorem}\label{th3.2}
Assume that $h\in \mathcal{C}(D)$,  $F\in \mathcal{C}(\br)$ and locally
Lipschitz in $\br$, then there exists a unique solution of the
problem \ref{pbb} locally in time which can be extended globally in
time.
  \end{theorem}
\begin{proof}
We know from Theorem  {\rm(\ref{th2.1})}  that, to prove the
existence and uniqueness of the solution {\rm(\ref{intSol})} of
 Problem {\rm(\ref{pbb})}, it is enough to solve the
Volterra integral
equation {\rm(\ref{SV})}. So we rewrite it as follows

\begin{eqnarray}\label{SV2}
V(y , t)= f(y , t) +\int_{0}^{t}g(y , \tau , V(y , \tau))d\tau
 \end{eqnarray}
with

\begin{eqnarray}\label{Phi}
f(y , t)= {1\over t(2\sqrt{\pi\, t})^{n} } \int_{\br^{+}}\xi
e^{-{\xi^{2}\over 4t}}\left(\int_{\br^{n-1}} e^{-{\|y-
\eta\|^{2}\over 4t}}h(\xi , \eta)d\eta\right)d\xi
\end{eqnarray}
and

\begin{eqnarray}\label{Psi}
 g(t , \tau , y , V(y , \tau))=-{2 (t-\tau)^{-n/2}
 \over (2\sqrt{\pi})^{n} }  \int_{\br^{n-1}}
F\left(
 {1\over \tau}\int_{0}^{\tau}V(\eta\ , s) ds\right)
e^{-{\|y- \eta\|^{2}\over
4(t-\tau)}}d\eta.
 \end{eqnarray}

So we have to check the conditions $H1$ to $H4$ in Theorem 1.1 page 87,
and $H5$ and $H6$ in Theorem 1.2 page 91   in \cite{MILLER}.

\noindent$\bullet$ The function $f$ is defined and continuous for
all $(y , t)\in \br^{n-1}\times\br^{+}$, so $H1$ holds.

\smallskip

\noindent$\bullet$ The function $g$ is measurable in $(t, \tau, y
, x)$ for $0\leq \tau \leq t < +\infty$, $x\in\br^{+}$, $y\in
\br^{n-1}$, and continuous in $x$ for all $(y , t ,
\tau)\in\br^{n-1}\times\br^{+}\times\br^{+}$, $g(y , t , \tau ,
x)=0$ if $\tau >t$, so here we need the continuity of

$$
V(\eta , \tau) \mapsto  F\left({1\over \tau} \int_{0}^{\tau}  V(\eta , s) ds\right),
$$
which follows from the hypothesis that 
$F\in C(\br)$. So $H2$ holds.

\smallskip

\noindent$\bullet$ For all $k >0$ and all bounded set $B$ in $\br$,
 we have
\begin{eqnarray*}
|g(y , t , \tau , X)| &\leq & {2\over (2\sqrt{\pi})^{n}} 
\sup_{X\in B}
|F(X)| (t-\tau)^{-\frac{n}{2}}
\int_{\br^{n-1}} e^{-\frac{\|y-\eta\|^{2}}
        {4(t-\tau)}}
d\eta
\nonumber\\
&\leq & {2\over (2\sqrt{\pi})^{n}} \sup_{X\in B}|F(X)|
(t-\tau)^{-\frac{n}{2}}(2\sqrt{\pi (t-\tau)})^{n-1}
\nonumber\\
&= & {1\over \sqrt{\pi}}  \sup_{X \in B} |F(X)|{1\over \sqrt{
(t-\tau)}}
\end{eqnarray*}
 thus there exists a measurable function $m$ given by
\begin{eqnarray}\label{m}
 m(t , \tau)={1\over \sqrt{\pi}}  \sup_{X
\in B} |F(X)|{1\over \sqrt{ (t-\tau)}}
\end{eqnarray}
such that
\begin{eqnarray}\label{H3}
|g(y , t , \tau , X)| \leq m(t , \tau) \quad \forall 0\leq \tau \leq
t\leq k,\quad X\in B
\end{eqnarray}
and satisfies
\begin{eqnarray*}
 \sup_{t\in [0 , k]}
 \int_{0}^{t} m(t , \tau) d\tau
&=& {1\over \sqrt{\pi}} \sup_{X\in B}|F(X)| \sup_{t\in [0 , k]}
 \int_{0}^{t}
  {1\over \sqrt{t-\tau}} d\tau
\nonumber\\
&= & {1\over \sqrt{\pi}}\sup_{X\in B}|F(X)| \sup_{t\in [0 ,
k]}\left(-2\sqrt{
  (t-\tau)}|_{0}^{t}\right)
\nonumber\\
  &=&{1\over \sqrt{\pi}}\sup_{X\in B}|F(X)| \sup_{t\in [0 , k]}2\sqrt{t}
  \leq {2\sqrt{k}\over \sqrt{\pi}} \sup_{X\in
  B}|F(X)| <\infty,
 \end{eqnarray*}
so  $H3$ holds.

\smallskip

\noindent$\bullet$ Moreover  we have also
\begin{eqnarray}\label{th1.2-p91}
\lim_{t \to 0^{+}}\int_{0}^{t}m(t , \tau)d\tau
&=&{1\over\sqrt{\pi}}\sup_{X \in B} |F(X)| \lim_{t \to 0^{+}}
   \int_{0}^{t}{d\tau\over
\sqrt{t-\tau}}
\nonumber\\
&=&  {1\over\sqrt{\pi}}\sup_{X \in B} |F(X)|\lim_{t \to
0^{+}}(2 \sqrt{t}) =0,
\end{eqnarray}
and
\begin{eqnarray}\label{th2.3-p97}
\lim_{t \to 0^{+}}\int_{-T}^{-T+t}m(t , \tau)d\tau =
{1\over\sqrt{\pi}}\sup_{X \in B} |F(X)|\lim_{t \to 0^{+}} 
2\left(\sqrt{t+T}-\sqrt{T}\right)
 =0.
\end{eqnarray}

\smallskip

\noindent$\bullet$ For each compact subinterval $J$ of $\br^{+}$,
each  bounded set $B$
 in $\br^{n-1}$, and each $t_{0}\in \br^{+}$, we set
\begin{eqnarray*}
\mathcal{A}(t, y , V(\eta)) =|g(t , \tau ; y , V(\eta , \tau))-g(t_{0}
, \tau ; y , V(\eta , \tau))|.
 \end{eqnarray*}

\begin{eqnarray*}
\mathcal{A}(t, y , V(\eta))={2\over (2\sqrt{\pi})^{n}}\int_{J}
\left|
 \int_{\br^{n-1}} 
 \left[
 {e^{-{\|y- z\|^{2}\over 4(t-\tau)}} 
  \over (t-\tau)^{{n/2}}}
 -
 {e^{-{\|y-z\|^{2}\over 4(t_{0}-\tau)}}
 \over (t_{0} -\tau)^{{n/2}}}
 \right]
 F\left(
 {1\over \tau}\int_{0}^{\tau}V(z \ , s) ds\right)
 dz \right| d\tau
 \end{eqnarray*}
as the function $\tau \mapsto V(z , \tau)$ is continuous then
$$  \tau \mapsto
 {1\over \tau}\int_{0}^{\tau}V(z , s) ds
 $$
is $\mathcal{C}^{1}(\br^{+})$
and  is
in the compact $B\subset \br$ for all $z\in \br^{n-1}$, so by the
continuity of $F$ we get $F\left(
 {1\over \tau}\int_{0}^{\tau}V(z\ , s) ds\right)\subset F(B)$, that is
there exists $M>0$ such that 
$
\left|F\left(
 {1\over \tau}\int_{0}^{\tau}V(z\ , s) ds\right)
\right|\leq M$ for all
$(z, \tau)\in\br^{n-1}\times \br^{+}$. So

\begin{eqnarray*}
\mathcal{A}(t, y , V(\eta)) \leq
{2M\over (2\sqrt{\pi})^{n}} 
 \left|
  \int_{\br^{n-1}}
{e^{-{\|y-z\|^{2}\over 4(t-\tau)}}
\over\sqrt{(t-\tau)}^{n}}dz
  -
 \int_{\br^{n-1}}
{e^{-{\|y-z\|^{2}\over
4(t_{0}-\tau)}}\over\sqrt{(t_{0}-\tau)}^{n}} dz \right|
 \end{eqnarray*}
using
\begin{eqnarray*}\label{for18}
 \int_{\br^{n-1}}  \exp\left[-{\|y-z\|^{2}\over
       4(t-\tau)}\right]dz =  \left(2\sqrt{\pi(t-\tau)}\right)^{n-1}
             \end{eqnarray*}
 we obtain
\begin{eqnarray*}
\mathcal{A}(t, y , V(\eta)) &\leq&
{2M\over (2\sqrt{\pi})^{n}} 
 \left|
{ (2\sqrt{\pi(t-\tau)})^{n-1}  \over (\sqrt{t-\tau})^{n}}
  -
{ (2\sqrt{\pi(t_{0}-\tau)})^{n-1} \over
(\sqrt{t_{0}-\tau})^{n}}\right|
\nonumber\\
&\leq&
{M\over \sqrt{\pi}} 
 \left|
{{\sqrt{t_{0}-\tau} -  \sqrt{t-\tau}}
 \over
\sqrt{(t-\tau)(t_{0}-\tau)} }\right|.
 \end{eqnarray*}

Thus we deduce that
\begin{eqnarray*}
 \lim_{t\to t_{0}}  \int_{J}
 \sup_{V(\eta)\in \mathcal{C}(J , B)}\mathcal{A}(t, y , V(\eta))d\eta =0.
 \end{eqnarray*}
So  $H4$ holds.

\smallskip

\noindent$\bullet$ For all compact $I\subset \br^{+}$, for all
function $\psi \in \mathcal{C}(I ,\br^{n})$, and all $t_{0}>0$,
\begin{eqnarray*}
&&|g(t , \tau ; \psi(\tau))-g(t_{0} , \tau ,\psi(\tau))|=
\nonumber\\
&&{2\over
(2\sqrt{\pi})^{n}}\left|
 \int_{\br^{n-1}}
  F\left(
 {1\over \tau}\int_{0}^{\tau}\psi(s) ds\right)
 \left(
 {e^{-{\|y- \eta\|^{2}\over 4(t-\tau)}}
 \over (t-\tau)^{{n/2}}}
 -
 {e^{-{\|y-\eta\|^{2}\over 4(t_{0}-\tau)}}\over
(t_{0} -\tau)^{{n/2}}} \right)d\eta \right|
 \end{eqnarray*}
as $F\in \mathcal{C}(\br)$ and $\psi \in \mathcal{C}(I ,\br^{n})$ then
there exists a constant $M>0$ such that
$$\left|
F\left(
 {1\over \tau}\int_{0}^{\tau}\psi(s) ds\right)
\right|\leq M,  \quad  \forall \tau \in I.$$
Then we obtain as for H4, that
\begin{eqnarray*}
\lim_{t\to t_{0}}\int_{I}|g(t , \tau ; \psi(\tau))-g(t_{0} , \tau
,\psi(\tau))|d\tau =0.
 \end{eqnarray*}
 So H5 holds.

\noindent$\bullet$ Now for each constant $k>0$ and each bounded set
$B\subset \br^{n-1}$ there exists a measurable function $\varphi$
such that
\begin{eqnarray*}
|g(y , t , \tau , x) - g(y , t , \tau , X)|\leq \varphi(t ,
\tau)|x-X|
 \end{eqnarray*}
 whenever $0\leq \tau \leq t \leq k$ and both $x$ and $X$ are in
 $B$. Indeed as $F$ is assumed locally Lipschitz function in $\br$
 there exists constant $L>0$ such that
\begin{eqnarray*}
\left|F\left(
 {1\over \tau}\int_{0}^{\tau}x(s) ds\right)
 -
 F\left(
 {1\over \tau}\int_{0}^{\tau}X(s) ds\right)
 \right|
 &\leq& L
 \left|
  {1\over \tau}\int_{0}^{\tau}x(s) ds
 -
 {1\over \tau}\int_{0}^{\tau}X(s) ds
  \right|
  \nonumber\\
  &&
 \leq L|x - X|
  \quad \forall (x , X)\in B^{2}
 \end{eqnarray*}
   then we have
\begin{eqnarray*}
|g(y , t , \tau , x) - g(y , t , \tau , X)|&\leq&
 {2L\over (2\sqrt{\pi})^{n}}\left(
 \int_{\br^{n-1}} e^{-{\|y- \eta\|^{2}\over 4(t-\tau)}}d\eta \right)
  (t-\tau)^{-{n/2}}
|x-X|
\nonumber\\
 &\leq&
 {L \over \sqrt{\pi(t-\tau)}}|x-X|,
 \end{eqnarray*}

 then $\varphi(t ,\tau)=
 {L\over \sqrt{\pi(t-\tau)}}$. We have also for each $t\in [0 , k]$
 the function  $\varphi\in L^{1}(0 , t)$ as a function of $\tau$ and
 we have also

\begin{eqnarray*}
 \int_{t}^{t+l}\varphi(t +l, \tau) d\tau
 = {L\over \sqrt{\pi}}  \int_{t}^{t+l}
 {d\tau\over \sqrt{t+l-\tau}}
 ={L\over \sqrt{\pi}} (2\sqrt{l})\to 0 \quad \mbox{with }
 l\to 0.
\end{eqnarray*}

  So H6 holds. All the conditions H1 to H6 are satisfied with
  (\rm{\ref{th1.2-p91})} and (\rm{\ref{th2.3-p97})}.

Thus from \cite{MILLER} (Theorem 1.1 page 87, Theorem 1.2 page 91
and  Theorem 2.3 page 97) there exists a unique solution, local in time,
to the Volterra integral equation {\rm(\ref{Volera})} which can be
extended globally in time.
  Then the proof of this theorem is complete.
  \end{proof}

\bigskip

  \section{The  one-dimensional case of Problem \ref{pbb}}
   \label{section D1}

       Let us consider now the one-dimensional case of Problem \ref{pbb}  for  the temperature defined by

 \begin{problem}\label{pb1d}
 Find the temperature $u$ at $(x , t)$ such that it satisfies the
 following conditions
 \begin{eqnarray*}\label{eqchP1}
 u_{t} - u_{xx} &=& -F\left({1\over t}\int_{0}^{t}
 u_{x}(0, s)ds\right), \qquad x>0, \quad
      t>0, \label{cpb1}\nonumber\\
            u(0,  t)&=& 0, \quad  t>0, \nonumber\\
  u(x, 0)&=& h(x),   \qquad x>0. \quad\label{cNIpb}
    \end{eqnarray*}
\end{problem}

Taking into account that
\begin{eqnarray}
 \int_{0}^{t} G(x , t, \xi, \tau) d\xi = erf\left({x\over 2\sqrt{t-\tau}}\right)
\end{eqnarray}
thus the solution of the {\it Problem }\ref{pb1d} is given by

\begin{eqnarray}
 u(x , t) = u_{0}(x, t) - \int_{0}^{t}
 erf\left({x\over 2\sqrt{t-\tau}}\right)
 F\left({1\over\tau }\int_{0}^{\tau} V(\sigma)d\sigma \right) d\tau
\end{eqnarray}
with
\begin{eqnarray}\label{u}
 u_{0}(x, t) = \int_{0}^{t} G(x, t, \xi, 0) h(\xi) d\xi
\end{eqnarray}
and $V(t)= u_{x}(0, t)$ is the the solution of the following
Volterra integral equation of the second kind
\begin{eqnarray}\label{Volterra}
 V(t)= V_{0}(t)
-
\int_{0}^{t} {F\left({1\over\tau }
  \int_{0}^{\tau} V(\sigma)d\sigma\right)
\over\sqrt{\pi(t-\tau)}} d\tau
\end{eqnarray}
where
\begin{eqnarray}\label{V0}
 V_{0}(t)= {1\over 2\sqrt{\pi}t^{3/2}}\int_{0}^{+\infty} \xi e^{-\xi^{2}/4t}h(\xi) d\xi
 = {2\over \sqrt{\pi t}}\int_{0}^{+\infty} \eta e^{-\eta^{2}}h(2\sqrt{t} \,\eta) d\eta.
\end{eqnarray}

For the particular case

\begin{eqnarray}\label{h0}
 h(x)=h_{0} >0  \mbox{ pour } x>0, \mbox{ and } F(V)= \lambda V  \mbox{ pour } \lambda\in \br
\end{eqnarray}

then we have
\begin{eqnarray}\label{U_{0}}
u_{0}(t , x)= h_{0} erf\left({x\over 2\sqrt{t}}\right)
\end{eqnarray}

and the integral equation (\ref{Volterra}) becomes
\begin{eqnarray}\label{Vol1n}
 V(t) = {h_{0}\over \sqrt{\pi t}}
 - \lambda \int_{0}^{t}{{1\over\tau } \int_{0}^{\tau} V(\sigma) d\sigma\over \sqrt{\pi(t-\tau)}}d\tau.
\end{eqnarray}

Then, we have
\begin{eqnarray}\label{u39}
 u(x , t) = h_{0} erf\left({x\over 2\sqrt{t}}\right)
 - \lambda \int_{0}^{t}  erf\left({x\over 2\sqrt{t-\tau}}\right) W(\tau) d\tau
\end{eqnarray}
where $W(t)$ is defined by
\begin{eqnarray}\label{V310}
 W(t)= {1\over t}\int_{0}^{t} V(\tau) d\tau
 = {1\over t}\int_{0}^{t} u_{x}(0, \tau) d\tau.
\end{eqnarray}

By using the integral equation (\ref{Vol1n}) for $V(t)$ we obtain
for $W(t)$ the following Volterra integral equation of the second
kind:

\begin{eqnarray}\label{V311}
 W(t) &=&  {1\over t}\int_{0}^{t} \left[{h_{0}\over \sqrt{\pi \tau}}
 - \lambda \int_{0}^{\tau}  {{1\over \mu}\int_{0}^{\mu} V(\sigma) d\sigma\over \sqrt{\pi(\tau-\mu)}} d\mu\right] d\tau
 \nonumber\\
 &=&  {1\over t} \left[ 2h_{0}\sqrt{t\over \pi}
 - \lambda \int_{0}^{t}\left[\int_{0}^{\tau} {W(\mu)\over \sqrt{\pi(\tau-\mu)}} d\mu\right] d\tau\right]
  \nonumber\\
 &=& { 2h_{0}\over\sqrt{ \pi t}}- {\lambda\over t}
 \int_{0}^{t}\left[\int_{\mu}^{t} {W(\mu)\over \sqrt{\pi(\tau-\mu)}} d\tau\right] d\mu
  \nonumber\\
 &=& { 2h_{0}\over\sqrt{ \pi t}}
 -  {2\lambda\over \sqrt{ \pi }} {1\over t} \int_{0}^{t}W(\tau) \sqrt{t-\tau} \, d\tau,  \quad  t>0,
\end{eqnarray}
by using that
$$ \int_{\mu}^{t} {d\tau\over \sqrt{\tau-\mu}}= 2 \sqrt{t-\mu}.$$

Therefore, we deduce the following results
\begin{theorem}\label{th3.1}
Taking $h$ and $F$ as in {\rm(\ref{h0})}, 
the solution of the non-classical heat conduction {\it Problem }{\rm{\ref{pb1d}}} is given by {\rm(\ref{u39})} where
 $W(t)$ is the solution of the Volterra integral equation {\rm(\ref{V311})}.
 Moreover, its  Laplace transform $\mathcal{L}$ is given by the following
 expression:
 \begin{eqnarray}\label{312}
  Q(s) =\mathcal{L}(W(t))(s)= \frac{h_{0}}{\lambda}
        \left(1- e^{-\frac{2\lambda}{\sqrt{s}}}\right),
 \end{eqnarray}
and $W(t)$is given by the following difference of two series  with infinite radii of convergence:

  \begin{eqnarray}\label{313}
   W(t)&=& \frac{2h_{0}}{\sqrt{\pi  t}}
   \left(  1+ \sum_{n=1}^{+\infty} 
        \frac{(4\lambda^{2} t)^{n}}
             {(2n+1)n! [(2n-1)!!]^{2}}  \right)
   \nonumber\\
   &&-2h_{0}\lambda \left(1+ 
   \sum_{n=1}^{+\infty} 
   \frac{(2\lambda^{2} t)^{n}}{(n+1)(n!)^{2} (2n+1)!!}\right)
  \end{eqnarray}
\end{theorem}

\begin{proof}
 By using the integral equation (\ref{V311}) for the real function $W(t)$, the Laplace transform $Q(s)$ of $W(t)$
 satisfies the following first order ordinary differential problem

  \begin{eqnarray}\label{314}
  && Q'(s)- {\lambda\over s^{3/2}} Q(s) = - {h_{0}\over s^{3/2}}, \quad \Re(s)>0
   \nonumber\\
  && Q(+\infty)=0,
  \end{eqnarray}
whose solution is given by (\ref{312}). From a series development of
the exponential function we obtain

 \begin{eqnarray*}
Q(s)&=&  {h_{0}\over \lambda} \sum_{n=1}^{+\infty} {(-1)^{n+1}\over n!} {(2\lambda)^{n}\over s^{n/2}}
\nonumber\\
&=&  {h_{0}\over \lambda}\left( \sum_{k=0}^{+\infty}{(2\lambda)^{2k+1}\over (2k+1)!  \,  s^{k+ {1\over2}}}
      - \sum_{k=1}^{+\infty}{(2\lambda)^{2k}\over (2k)! \,  s^{k}}\right)
   \nonumber\\
&=&   2   h_{0}\left( {1\over  s^{1\over2}}
+  \sum_{k=1}^{+\infty}{(2\lambda)^{2k}\over (2k+1)(2k)!  \,  s^{k+ {1\over2}}}
      - \sum_{k=1}^{+\infty}{2^{k-1}\lambda^{2k-1}\over k! (2k-1)!! \,  s^{k}}\right)
  \end{eqnarray*}
and therefore we get

 \begin{eqnarray*}
  W(t)= \mathcal{L}^{-1}(Q(s))(t) &=& 2   h_{0}\left( {1\over \sqrt{\pi t}}
      +   {1\over \sqrt{\pi t}}\sum_{n=1}^{+\infty}
      {(2\lambda)^{2n} \,  2^{n} \,  t^{n}\over (2n+1)(2n)! (2n-1)!!}\right)
         \nonumber\\
&-&  2 h_{0}\left(\lambda  + \sum_{n=1}^{+\infty}
      {2^{n} \lambda^{2n+1}  \,  t^{n}\over (n+1)! n! (2n+1)!!}\right)
 \end{eqnarray*}
that is the expression (\ref{313}) for $W(t)$ holds  by using that
$$ \mathcal{L}^{-1}\left( {1\over s^{{1\over 2}} }\right)(t) = {1\over \sqrt{\pi t}},  \qquad
\mathcal{L}^{-1}\left({ 1\over s^{n} }\right)(t) = {t^{n-1}\over (n-1)!},   $$
$$\mathcal{L}^{-1}\left({1\over  s^{n+{1\over 2}}}\right)(t) =  { 2^{n} \, t^{n-{1\over 2}} \over (2n-1)!! \sqrt{\pi}}, \quad  n\geq 1$$

$$  (2n)! = 2^{n} n! (2n-1)!!$$

and the definition
$$  (2n-1)!! = (2n-1)(2n-3)\cdots 5\cdot 3\cdot 1.$$
\end{proof}

 \begin{corollary}
  The heat flux at the boundary $x=0$ of the solution of the {\it Problem} {\rm{\ref{pb1d}}} is given by
   \begin{eqnarray}\label{315}
    u_{x}(0, t) = {h_{0}\over \sqrt{\pi t}} - {\lambda\over \sqrt{\pi}}\int_{0}^{t} {W(\tau)\over \sqrt{t-\tau}} d\tau
   \end{eqnarray}
where $W(t)$ is given by {\rm(\ref{313})}.
 \end{corollary}

 \begin{corollary}
  The first terms of the development of the serie {\rm(\ref{313})} of the average of the total heat flux at $x=0$
  are given by
 $$2h_{0}\left( {1\over \sqrt{\pi t}} - \lambda+ {4\lambda^{2}\over 3\sqrt{\pi}}\sqrt{t}- {\lambda^{3} \over 3} t
 + {8 \lambda^{4}  \over 45 \sqrt{\pi}} t^{3/2} - {\lambda^{5}\over 45} t^{2} \right) $$
  which give us a singularity of $W$ of the type $t^{-1/2}$ at $t=0$.

 Moreover, the first terms of the development of the heat flux at $x=0$ of the expression {\rm(\ref{315})} are given by
$$ h_{0}\left( {1\over \sqrt{\pi t}} - {\lambda\over 4}+ {4\lambda^{2}\over \sqrt{\pi}}\sqrt{t}- {4\lambda^{3} \over 3} t
 + {8 \lambda^{4}  \over 9\sqrt{\pi}} t^{3/2} - {2\lambda^{5}\over 15} t^{2} + {32 \lambda^{6}\over 675 \sqrt{\pi}} t^{5/2}\right)$$
  which give also us a singularity of $u_{x}(0, t)$ of the type $t^{-1/2}$ at $t=0$.
 \end{corollary}
\begin{proof}
 It follows from the following results:

 \begin{eqnarray*}
  \int_{0}^{t} {d\tau\over \sqrt{\tau(t-\tau)}} = {\pi\over 8},
  \qquad
  \int_{0}^{t} {d\tau\over \sqrt{t-\tau}} = 2\sqrt{t},
 \nonumber\\
  \int_{0}^{t} {\sqrt{\tau}\over \sqrt{t-\tau}} d\tau = {\pi\over 2}t,
  \qquad
   \int_{0}^{t} {\tau\over \sqrt{t-\tau}} d\tau = {4\over 3} t^{3/2},
    \nonumber\\
     \int_{0}^{t} {\tau^{3/2}\over \sqrt{t-\tau}} d\tau =  {3\pi\over 8} t^{2},
      \qquad
    \int_{0}^{t} {\tau^{2}\over \sqrt{t-\tau}} d\tau = {16\over 15} t^{5/2},
 \end{eqnarray*}
which can be generalized to
 \begin{eqnarray*}
 \int_{0}^{t} {\tau^{n}\over \sqrt{t-\tau}} d\tau =  {(2n)!\over [(2n-1)!!]^{2}} {t^{n+ {1\over 2}}\over n+ {1\over 2}},
 \quad n\geq 1,
      \nonumber\\
    \int_{0}^{t} {\tau^{n- {1\over 2}}\over \sqrt{t-\tau}} d\tau = {\pi[(2n-1)!!]^{2}\over (2n)!} t^{n},
    \quad n\geq 1.
  \end{eqnarray*}
\end{proof}

Now, we will give a new proof of the serie (\ref{313}) for the
average of the total flux $W(t)$. We use the Adomian decomposition
method \cite{Adm, Adl, Adm1, MB-DT1, waz1, waz3} through a serie
expansion  of the type
 \begin{eqnarray}\label{serie}
W(t)= \sum_{n=0}^{+\infty} W_{n}(t)
   \end{eqnarray}
 in the Volterra integral equation  (\ref{V311}). Taking
 \begin{eqnarray}\label{V0}
    W_{0}(t)= {2  h_{0}\over \sqrt{\pi t}}
   \end{eqnarray}
 we obtain  the following recurrence formulas
    \begin{eqnarray}\label{Vn}
     W_{n}(t)= -{2 \lambda \over t \sqrt{\pi}}\int_{0}^{t} W_{n-1}(\tau) \sqrt{t-\tau} \, d\tau, \quad n\geq 1.
    \end{eqnarray}
Then by (\ref{V0}) and (\ref{Vn}) we get 
 \begin{eqnarray}\label{V1}
   W_{1}(t)&=&-{2 \lambda \over t \sqrt{\pi}}\int_{0}^{t} W_{0}(\tau) \sqrt{t-\tau} \, d\tau
  \nonumber\\
  &=& -{4 \lambda h_{0}\over t \pi}\int_{0}^{t}  {\sqrt{t-\tau}\over \sqrt{\tau}} d\tau = -2\lambda h_{0}. 
 \end{eqnarray}

 \begin{theorem}
  Moreover by a double induction principle we have
     \begin{eqnarray}\label{V2n}
    W_{2n}(t)= {2  h_{0}\over \sqrt{\pi t}} {(4\lambda^{2} t)^{n}\over (2n+1) n! [(2n-1)!!]^{2}}, \quad  n\geq 1,
     \end{eqnarray}
and
   \begin{eqnarray}\label{V2n+1}
     W_{2n+1}(t)= - 2 \lambda h_{0} {(2\lambda^{2} t)^{n}\over (n+1) (n!)^{2} (2n+1)!!}, \quad  n\geq 1,
   \end{eqnarray}
with $W_{0}$ and  $W_{1}$ are given respectively by {\rm(\ref{V0})} and  {\rm(\ref{V1})}.
 \end{theorem}

\begin{proof}
 Using   (\ref{V0}) and (\ref{Vn})
we get

 \begin{eqnarray*}
   W_{2}(t)&=&-{2 \lambda \over t \sqrt{\pi}}\int_{0}^{t} W_{1}(\tau) \sqrt{t-\tau} \, d\tau
  \nonumber\\
  &=& {4 \lambda^{2} h_{0}\over t \sqrt{\pi}}\int_{0}^{t}  \sqrt{t-\tau} d\tau = {8\lambda^{2} h_{0}\over 3\sqrt{\pi}}\sqrt{t},
 \end{eqnarray*}

\begin{eqnarray*}
   W_{3}(t)&=&-{2 \lambda \over t \sqrt{\pi}}\int_{0}^{t} W_{2}(\tau) \sqrt{t-\tau} \, d\tau
  \nonumber\\
  &=& -{16 \lambda^{3} h_{0}\over 3t \pi}\int_{0}^{t}  \sqrt{\tau} \sqrt{t-\tau} d\tau
  = -{2\lambda^{3} h_{0}\over 3}t,
 \end{eqnarray*}

\begin{eqnarray*}
   W_{4}(t)&=&-{2 \lambda \over t \sqrt{\pi}}\int_{0}^{t} W_{3}(\tau) \sqrt{t-\tau} \, d\tau
  \nonumber\\
  &=& {4 \lambda^{4} h_{0}\over 3t \sqrt{\pi}}\int_{0}^{t}  \tau \sqrt{t-\tau} d\tau
  = {16\lambda^{4} h_{0}\over 45 \sqrt{\pi} }t^{3/2},
 \end{eqnarray*}

taking into account that
\begin{eqnarray*}
 \int_{0}^{t}  {\sqrt{t-\tau}\over \sqrt{\tau}} d\tau ={\pi\over 2}t,  \qquad
 \int_{0}^{t}  \sqrt{t-\tau} d\tau ={2\over 3}t^{3/2},
   \nonumber\\
\int_{0}^{t}  \sqrt{\tau} \sqrt{t-\tau} d\tau = {\pi \over 8} t^{2},
  \qquad
  \int_{0}^{t}  \tau \sqrt{t-\tau} d\tau = {4\over 15} t^{5/2},
\end{eqnarray*}

and their generalizations by
\begin{eqnarray}\label{f}
 \int_{0}^{t}  \tau^{n} \sqrt{t-\tau} d\tau 
 = {2^{n+1} n! t^{n+\frac{3}{2}} \over (2n+3)!!}, \quad n\geq 1,
\end{eqnarray}
\begin{eqnarray}\label{g}
  \int_{0}^{t}  \tau^{n-{1\over 2}} \sqrt{t-\tau} d\tau = {\pi (2n-1)!! t^{n+1} \over 2^{n+1} (n+1)!  }, \quad n\geq 1.
\end{eqnarray}

The first step of the double induction principle is verified taking
into account the above computations. For the second step, we suppose
by induction hypothesis that we have (\ref{V2n})  and (\ref{V2n+1}).
Therefore,  we obtain

\begin{eqnarray}
 W_{2n+2}(t) &=&{-2\lambda\over t\sqrt{\pi}}\int_{0}^{t} W_{2n+1}(\tau)\sqrt{t-\tau} d\tau
 \nonumber\\
 &=& {4\lambda^{2}h_{0}\over  t\sqrt{\pi}}{ (2\lambda^{2})^{n}\over (n+1)(n!)^{2}(2n+1)!!} \int_{0}^{t}
 \tau^{n} \sqrt{t-\tau} d\tau
  \nonumber\\
 &=& {2 h_{0}\over \sqrt{\pi t}}{ (4\lambda^{2} t)^{n+1}\over (2n+3)(n+1)! [(2n+1)!!]^{2}},
\end{eqnarray}

and

\begin{eqnarray}
 W_{2n+3}(t) &=&{-2\lambda\over t\sqrt{\pi}}\int_{0}^{t} W_{2n+2}(\tau)\sqrt{t-\tau} d\tau
 \nonumber\\
 &=& - {4 \lambda h_{0}\over  t \pi}{ (4\lambda^{2})^{n+1}\over (2n+3)(n+1)![(2n+1)!!]^{2}} \int_{0}^{t}
 \tau^{n+{1\over 2}} \sqrt{t-\tau} d\tau
  \nonumber\\
 &=& -2\lambda h_{0}  { (2\lambda^{2} t)^{n+1}\over (n+2)[(n+1)!]^{2}(2n+3)!! },
\end{eqnarray}
by using (\ref{f}) and (\ref{g}). Then, the proof by the induction principle  holds.
\end{proof}

  \bigskip
\noindent{\bf Conclusion:} We have obtained the global solution of a
non-classical heat conduction problem in a semi-n-dimensional space,
in which  the source depends of the average of the total heat flux
on the face $x=0$. Moreover, for the one-dimensional case  we have
obtained the explicit solution by using the Laplace transform and
also the Adomian decomposition method.

\bigskip

\bigskip

\noindent{\bf Acknowledgements:}
 This paper was partially sponsored by the Institut Camille Jordan St-Etienne University for first author,
 and the projects PIP $\#$ 0275 from CONICET - Univ. Austral and ANPCyT PICTO
 Austral 2016 $\#$ 090 (Rosario, Argentina) for the second author.


    \end{document}